\documentclass[a4paper,11pt]{article}

\usepackage[utf8]{inputenc}
\usepackage[T1]{fontenc}
\usepackage[english]{babel}

\usepackage[pdftex]{graphicx}
\usepackage[usenames,dvipsnames]{xcolor}

\usepackage{booktabs}
\usepackage{fullpage}

\usepackage{amsmath,amssymb,amsfonts,amsthm}
\usepackage{mathabx}
\usepackage[ruled,linesnumbered]{algorithm2e}

\newtheorem{lemma}{Lemma}
\newtheorem{theorem}{Theorem}

\newtheorem{conjecture}{Conjecture}
\newtheorem{proposition}{Proposition}

\theoremstyle{definition}
\newtheorem{definition}{Definition}

\usepackage{cleveref}

\usepackage{enumitem}

\newcommand{\eps}{\varepsilon}

\def\poly{\operatorname{poly}}
\newcommand{\ceil}[1]{\left\lceil{#1}\right\rceil}
\newcommand{\floor}[1]{\left\lfloor{#1}\right\rfloor}
\def\deg{\operatorname{deg}}
\newcommand{\abs}[1]{\left | #1 \right |}
\newcommand{\R}{\mathbb{R}}
\def\supp{\operatorname{supp}}
\newcommand{\norm}[1]{\left \lVert #1 \right \rVert}
\newcommand{\snorm}[1]{\norm{#1}_{\phi}}

\newcommand{\set}[1]{\left \{ #1 \right \}}

\usepackage{authblk}

\title{Finding Even Cycles Faster via Capped $k$-Walks}

\author{Søren Dahlgaard\thanks{Research partly supported by Advanced Grant
DFF-0602-02499B from the Danish Council for Independent Research}}
\author{Mathias Bæk Tejs Knudsen\thanks{Research partly supported by Advanced
Grant DFF-0602-02499B from the Danish Council for Independent and the FNU
project AlgoDisc -- Discrete Mathematics, Algorithms, and Data Structures.}}
\author{Morten Stöckel\thanks{Research partly supported by Villum Fonden and
the DABAI project.}}
\affil{University of Copenhagen\\\texttt{[soerend,knudsen,most]@di.ku.dk}}
\date{}

\begin{document}

\maketitle
\begin{abstract}
    Finding cycles in graphs is a fundamental problem in algorithmic graph
    theory. In this paper, we consider the problem of finding and reporting a
    cycle of length $2k$ in an undirected graph $G$ with $n$ nodes and $m$ edges for
    constant $k\ge 2$. A classic result by Bondy and Simonovits [J.
    Combinatorial Theory, 1974] implies that if $m \geq 100k n^{1+1/k}$, then
    $G$ contains a $2k$-cycle, further implying that one needs to consider only
    graphs with $m = O(n^{1+1/k})$.

    Previously the best known algorithms were an $O(n^2)$ algorithm due to
    Yuster and Zwick [J. Discrete Math 1997] as well as a
    $O(m^{2-(1+\ceil{k/2}^{-1})/(k+1)})$ algorithm by Alon et. al.
    [Algorithmica 1997].

    We present an algorithm that uses $O\left ( m^{2k/(k+1)} \right )$ time and
    finds a $2k$-cycle if one exists. This bound is $O(n^2)$ exactly when $m =
    \Theta(n^{1+1/k})$. When finding $4$-cycles our new bound coincides with
    Alon et. al., while for every $k>2$ our new bound yields a polynomial
    improvement in $m$.

    Yuster and Zwick noted that it is ``plausible to conjecture that $O(n^2)$
    is the best possible bound in terms of $n$''. We show ``conditional
    optimality'': if this hypothesis holds then our $O(m^{2k/(k+1)})$ algorithm
    is tight as well. Furthermore, a folklore reduction implies that no
    \emph{combinatorial} algorithm can determine if a graph contains a
    $6$-cycle in time $O(m^{3/2-\eps})$ for any $\eps>0$ unless boolean matrix
    multiplication can be solved combinatorially in time $O(n^{3-\eps'})$ for
    some $\eps' > 0$, which is widely believed to be false. Coupled with our
    main result, this gives tight bounds for finding $6$-cycles combinatorially
    and also separates the complexity of finding $4$- and $6$-cycles giving
    evidence that the exponent of $m$ in the running time should indeed
    increase with $k$.

    The key ingredient in our algorithm is a new notion of \emph{capped
    $k$-walks}, which are walks of length $k$ that visit only nodes according
    to a fixed ordering. Our main technical contribution is an involved
    analysis proving several properties of such walks which may be of
    independent interest.
\end{abstract}

\section{Introduction}
We study a basic problem in algorithmic graph theory.
Namely, given an undirected and unweighted graph $G=(V,E)$ and an
integer $\ell$, does $G$ contain a cycle of length exactly $\ell$ (denoted
$C_\ell$)? If a $C_\ell$ exists, we would also like the algorithm to return
such a cycle.
As a special case, when $\ell = n$ is the number of nodes in the graph, we
are faced with the well-known problem of finding a hamiltonian cycle, which was
one of Karp's original 21 NP-complete problems \cite{Karp1972}.
In fact, the problem is NP-complete when $\ell = n^{\Omega(1)}$.

On the other end of the spectrum, when $\ell = O(1)$ is a constant, the
problem is in FPT\footnote{Informally, a problem of size $n$ parameterized
by $k$ is in FPT if it can be solved in time $f(k)\cdot n^{O(1)}$, where $f$ is
a function independent of $n$.}
as first shown by Monien in
1985~\cite{Monien85}, by giving an $O(f(\ell)\cdot m)$ algorithm to
determine if any given node $u$ is contained in a $C_{\ell}$. For $\ell = 3$,
this is the classical problem of triangle-finding, which can be done in
$O(n^\omega)$ time using matrix multiplication, where $\omega < 2.373$ is the
matrix multiplication exponent~\cite{LeGall:2014:PTF:2608628.2608664}. This can
be extended to finding a $C_\ell$ for any constant $\ell=O(1)$ in time
$O(n^\omega)$ expected and $O(n^\omega\log n)$
deterministically~\cite{AlonYZ95}. When $\ell$ is odd, this is the fastest
known algorithm, however for even $\ell = 2k = O(1)$ one can do better.
To appreciate the difference, we must first understand
the following basic graph theoretic result about even cycles: Bondy and
Simonovits~\cite{BONDY197497} showed that if a graph with $n$ nodes has more
than $100k n^{1+1/k}$ edges, then the graph contains a $C_{2k}$.
In contrast, a graph on $n$ nodes can have $\Theta(n^2)$ edges without containing
any odd cycle, e.g. $K_{\floor{n/2},\ceil{n/2}}$. Using this lemma of Bondy
and Simonovits, it was shown by Yuster and Zwick~\cite{YusterZ97} how to find a
$C_{2k}$ for constant $k$ in time $O(n^2)$. They note that \emph{``it seems
plausible to conjecture that $O(n^2)$ is the best possible bound in terms of
$n$''}. Furthermore, when $m\ge 100k\cdot n^{1+1/k}$
we can use the algorithm of Yuster Zwick~\cite{YusterZ97} to find a $C_{2k}$ in
$O(n)$ expected time.
Given this situation, we seek an algorithm with a running time
$O(m^{c_k})$, which utilizes the sparseness of the graph, when $m$ is less than
$100k\cdot n^{1+1/k}$. By the above discussion, such an algorithm can
be turned into a $O(n^{c_k(1+1/k)})$ time algorithm for finding a
$C_{2k}$.
Therefore, if we believe that $O(n^2)$ indeed is the correct running time in
terms of $n$, we must also believe that the best possible value for $c_k$
is $2 - 2/(k+1)$. This is further discussed in Section~\ref{sec:hard_intro}
below.
Our main result is to present an algorithm which obtains exactly this
running time in terms of $m$ and $k$ for finding a $C_{2k}$. We show the
following.
\begin{theorem}\label{thm:2k_cycle}
    Let $G$ be an unweighted and undirected graph with $n$ nodes and $m$ edges, and let $k \geq 2$ be a positive integer.
    A $C_{2k}$ in $G$, if one exists, can be found in $O(k^{O(k)}m^{\frac{2k}{k+1}})$.
\end{theorem}
Theorem~\ref{thm:2k_cycle} presents the first improvement in more than 20 years
over a result of Alon, et al.~\cite{AlonYZ97}, who gave an algorithm with
$c_k = 2 - (1 + \frac{1}{\ceil{k/2}})/(k+1)$, i.e.,
a running time of $O(m^{4/3})$ for $4$-cycles and $O(m^{13/8})$ for $6$-cycles.
For $4$-cycles we obtain the same bound with Theorem~\ref{thm:2k_cycle}, but
for any $k > 2$ our new bound presents a polynomial improvement. In fact
our algorithm for finding a $C_8$ is faster than the algorithm of Alon, et al.
for finding a $C_6$. A comparison with known algorithms is shown below in
Figure~\ref{fig:comparison}.
\begin{figure}[htbp]
    \centering
    \includegraphics[width=0.6\columnwidth]{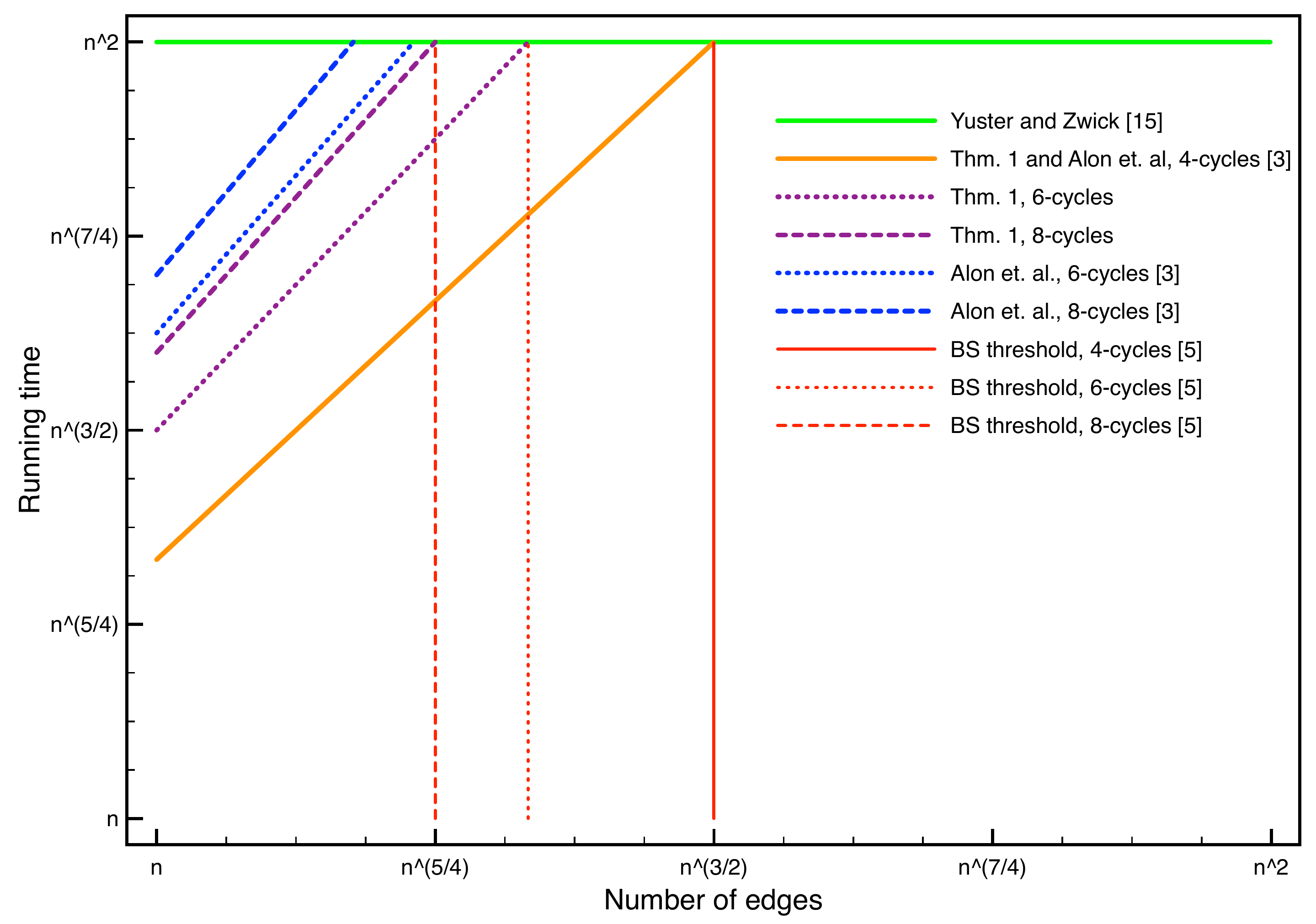}
    \caption{Comparisons of running times in terms of graph density. The
    illustration shows our algorithm from \Cref{thm:2k_cycle} compared to
    \cite{YusterZ97} and \cite{AlonYZ97}, and shows that it uses quadratic time
    exactly when the threshold from Bondy and Simonovits ensures the existence
    of a $2k$-cycle.}
    \label[figure]{fig:comparison}
\end{figure}

We present our algorithm as a black box reduction: Let $A$ be any algorithm
which can determine \emph{for a given node} $u$ if $u$ is contained in a
$C_{2k}$ in $O(f(k)\cdot m)$ time. Then our algorithm can transform $A$ into an
algorithm which finds a $C_{2k}$ in $O(g(k)\cdot m^{2k/(k+1)})$ time. Thus, one
may pick any such algorithm $A$ such as the original algorithm of
Monien~\cite{Monien85} or the seminal color-coding algorithm of Alon et
al.~\cite{AlonYZ95}. Our algorithm is conceptually simple, but the analysis is
technically involved and relies on a new understanding of the relationship
between the number of $k$-walks and the existence of a $C_{2k}$. By
introducing the notion of \emph{capped $k$-walks}, we show that an algorithm
enumerating all such capped $k$-walks starting in nodes with low degree will
either find a $2k$-cycle or spend at most $O(m^{2k/(k+1)})$ time. In some sense
this is a stronger version of the combinatorial lemma by Bondy and Simonovits,
as any graph with many edges must also have many capped $k$-walks.

\subsection{Hardness of finding cycles}\label{sec:hard_intro}
The literature on finding $\ell$-cycles is generally split into two kinds of
algorithms: \emph{combinatorial} and \emph{non-combinatorial} algorithms. Where
combinatorial algorithms (informally) are algorithms, which do not use the
structure of the underlying field and perform Strassen-like cancellation
tricks~\cite{Strassen69}. Interestingly, all known algorithms for finding
cycles of even length efficiently are combinatorial. There are several possible
explanations for this. One is that the hard instance for even cycles are
graphs, which are relatively sparse (i.e. $O(n^{1+1/k})$ edges), and in this
case it is difficult to utilize the power of fast matrix-multiplication.
Another is that matrix-multiplication based methods allows one to solve the
harder problem of directed graphs. Directed graphs are harder because we can no
longer make the guarantee that a $C_{2k}$ can always be found if the graph is
dense. Furthermore, a simple argument shows that the problem of finding a
$C_3$ can be reduced to the problem of finding a directed $C_\ell$ for any
$\ell > 3$. Especially this problem of finding a $C_3$ combinatorially has been
studied thoroughly in the line of work colloquially referred to as
\emph{Hardness in \textbf{P}}. This line of work is concerned with basing
hardness results on widely believed conjectures about problems in \textbf{P}
such as 3-SUM and APSP. One such popular conjecture (see
e.g.~\cite{AbboudW14,WilliamsW10}) is the combinatorial boolean matrix
multiplication (BMM) conjecture stated below.

\begin{conjecture}\label{conj:comb_bmm}
    There exists no \emph{combinatorial} algorithm for multiplying two $n\times
    n$ boolean matrices in time $O(n^{3-\eps})$ for any $\eps > 0$.
\end{conjecture}
It is known from \cite{WilliamsW10} that Conjecture~\ref{conj:comb_bmm} above
is equivalent to the statement that there exists no truly subcubic\footnote{An
algorithm running polynomially faster than cubic time, i.e. $O(n^{3-\eps})$ for
$\eps > 0$.} \emph{combinatorial} algorithm for finding a $C_3$ in graphs with
$n$ nodes and $\Theta(n^2)$ edges, and a simple reduction shows that this holds
for any odd $\ell \ge 3$. For even cycles, we show that a simple extension
to this folklore reduction gives the following result.
\begin{proposition}\label[proposition]{thm:clbs}
    Let $k \ge 3$ be a fixed integer with $k\ne 4$. Then there exists no
    \emph{combinatorial} algorithm that can find a $2k$-cycle in graphs with
    $n$ nodes and $m$ edges in time $O(m^{3/2-\eps})$ unless
    Conjecture~\ref{conj:comb_bmm} is false.
\end{proposition}
As noted, the proof of \Cref{thm:clbs} is a rather simple extension of the
reduction for odd cycles, but for completeness, we include the proof in
\Cref{sec:clb}. In particular, \Cref{thm:clbs} implies that our $O(m^{3/2})$
time algorithm for finding $6$-cycles is optimal among combinatorial
algorithms. Interestingly, \Cref{thm:clbs} also creates a separation between
finding $4$-cycles and finding larger even cycles, as both Alon, et
al.~\cite{AlonYZ97} and Theorem~\ref{thm:2k_cycle} provide an algorithm for
finding $4$-cycles in time $O(m^{4/3})$., which is polynomially smaller than
$O(m^{3/2})$. This gives evidence that a trade-off dependent on $k$ like the
one obtained in Theorem~\ref{thm:2k_cycle} is indeed necessary.

An important point of \Cref{thm:2k_cycle}, as mentioned earlier, is that it
is optimal if we believe that $\Theta(n^2)$ is the correct running time in
terms of $n$. This is formalized in the theorem below. Furthermore, we show
that \Cref{thm:2k_cycle} implies that any hardness result of $n^{2-o(1)}$ would
provide a link between the time complexity of an algorithm and the existence of
dense graphs without $2k$-cycles. A statement, which is reminiscent of the
Erdős Girth Conjecture.
\begin{theorem}\label{thm:n2opt}
    Let $k\ge 2$ be some fixed integer. For all $\eps > 0$ there exists
    $\delta > 0$ such that if no algorithm exists which
    can find a $C_{2k}$-cycle in graphs with $n$ nodes and $m$ edges in time
    $O(n^{2-\delta})$, then the following two statements
    hold.
    \begin{enumerate}
        \itemsep-2pt
        \item There is no algorithm which can detect if a graph contains a
            $C_{2k}$ in time $O(m^{2k/(k+1) - \eps})$.
        \item There exists an infinite family of graphs $\mathcal{G}$, such
            that each $G\in \mathcal{G}$ has $|E(G)| \ge |V(G)|^{1+1/k-\eps}$
            and contains no $C_{2k}$.
    \end{enumerate}
\end{theorem}

\subsection{Other results}
A problem related to that of finding a given $C_\ell$ is to determine the girth
(length of shortest cycle) of a graph $G$. In undirected
graphs, finding the shortest cycle in general can be done in time
$O(n^{\omega})$ time due to a seminal paper by Itai and Rodeh
\cite{DBLP:journals/siamcomp/ItaiR78}, and the shortest directed cycle can
be found using an extra factor of $O(\log n)$. In undirected graphs they also
show that a cycle that exceeds the shortest by at most one can be found in
$O(n^2)$ time. It was shown by Vassilevska Williams and
Williams~\cite{WilliamsW10} that computing the girth exactly is essentially as hard as boolean
matrix multiplication, that is, finding a combinatorial, truly subcubic
algorithm for computing the girth of a graph would break
Conjecture~\ref{conj:comb_bmm}. Thus, an interesting question is whether one
can approximate the girth faster, and in particular a main open question as
noted by Roditty and Vassilevska Williams~\cite{RodittyW12} is whether one can
find a $(2-\eps)$-approximation in $O(n^{2-\eps'})$ for any constants
$\eps,\eps'> 0$. They answered this question affirmatively for
\emph{triangle-free graphs} giving a $8/5$-approximation in $O(n^{1.968})$
time~\cite{RodittyW12}. By plugging Theorem~\ref{thm:2k_cycle} into their
framework we obtain the following result.
\begin{theorem}
    There exists an algorithm for computing a $8/5$-approximation of the girth
    in a triangle-free graph $G$ in time $O(n^{1.942})$.
\end{theorem}
\subsection{Capped $k$-walks}\label{sec:capped_walks}
The main ingredient in our analysis is a notion of \emph{capped $k$-walks}
defined below.
\begin{definition}
    Let $G = (V,E)$ be a graph and let $\preceq$ be a total ordering of
    $V$.
    For a positive integer, $k$, we say that a $(k+1)$-tuple $(x_0,\ldots,x_k) \in
    V^{k+1}$ is called a \emph{$\preceq$-capped $k$-walk} if $(x_0,\ldots,x_k)$ is
    a walk in $G$ and $x_0 \succeq x_i$ for each $i=1,2,\ldots,k$.
\end{definition}
When clear from the context we will refer to a $\preceq$-capped $k$-walk simply
by a capped $k$-walk. Our algorithm for finding $2k$-cycles essentially works
by enumerating all $\preceq$-capped $k$-walks (with some pruning applied),
where $\preceq$ is given by ordering nodes according to their degree. We will
show that by bounding the number of such $\preceq$-capped $k$-walks in graphs
with a not too large maximum degree, we obtain a bound on the running time of
our algorithm. Specifically, we show the following lemma.

\begin{lemma}\label[lemma]{lem:main_kwalks}
    Let $G = (V,E)$ be a graph, let $k$ be a positive integer, and assume that
    $G$ has maximum degree at most $m^{2/(k+1)}$. Let $\preceq$ be any ordering
    of the nodes in $G$ such that $u \preceq v$ for all pairs of nodes $u,v$
    such that $\deg(u) < \deg(v)$. If $G$ contains no $2k$-cycle, then the
    number of $\preceq$-capped $k$-walks is at most $f(k) m^{2k/(k+1)}$, where
    $f(k) = \left(O(k^2)\right)^{k-1} = k^{O(k)}$.
\end{lemma}

We also present a lower bound on the number of $\preceq$-capped $k$-walk, which
implies that graphs with a large number of edges contains a large number of
$\preceq$-capped $k$-walks.

\begin{lemma}\label[lemma]{lem:kwalks_lower}
    Let $G = (V,E)$ be a graph with $n$ nodes and $m$ edges. Let $\preceq$ be
    any ordering of $V$. The number of $\preceq$-capped $k$-walks is at least
        $n \cdot \left ( \frac{m}{2n} \right )^k$
\end{lemma}

Lemmas~\ref{lem:main_kwalks} and~\ref{lem:kwalks_lower} imply that
graphs with more than $Ck^2n^{1+1/k}$ edges and maximum degree at most
$m^{2/(k+1)}$ have a $2k$-cycle, for a
sufficiently large constant $C > 0$. Except from the extra factor of $k$
and the bound on the maximum degree, this shows that \Cref{lem:main_kwalks}
is stronger than the lemma of Bondy and Simonovits, which states that
graphs with at least than $100kn^{1+1/k}$ edges contain a $2k$-cycle. Indeed,
a graph with few edges may still contain many capped $k$-walks.

\subsection{Techniques and overview}
Our main technical contribution is the analysis of capped $k$-walks, outlined
in \Cref{sec:capped_walks} above. A standard way of reasoning about the number
of $k$-walks in a graph $G = (V,E)$ is to consider the adjacency matrix, $X_G$,
of $G$, where $X_G[i,j] = 1$ if $(i,j)\in E$ and $0$ otherwise. Here we denote the nodes of $G$ by $1, \ldots, n$.
Then the number of
$k$-walks in $G$ from $i$ to $j$ is exactly $X_G^k[i,j]$ and the total number
of $k$-walks is $\|X_G^k \mathbf{1}\|_1$, where $\mathbf{1} = (1,1,\ldots,1)^n$.
Furthermore the number of $k$-walks starting in a specific node $i$ is
$(X_G^k\mathbf{1})_i$. We will be interested in bounding the number of
$k$-walks starting in a specific subset $S\subseteq V$. This number can be
calculated as $\langle X_G^k\mathbf{1},\mathbf{1}_S\rangle$, where
$\mathbf{1}_S$ is the vector with $1$s in each index, $i$, such that $i\in S$
and $0$s elsewhere. Our goal will be to bound the norm of $X_G$ and use this to
bound the number of $k$-walks. However, bounding the $1$-norm leads to a too
large bound and cannot be used in proving Lemma~\ref{lem:main_kwalks}. We note that
the $1$-norm of a vector $v$, can be written as
\[
    \norm{v}_1 = \int_0^\infty \abs{\set{i\mid \abs{v_i}\ge x}} dx\ .
\]
We will instead consider the following related quantity, that we will call the
$\snorm{\cdot}$-norm.
\begin{definition}\label[definition]{def:snorm}
	For a vector $v \in \R^n$ we define the norm $\snorm{v}$ by
	\begin{align}
		\notag
		\snorm{v}
		=
		\int_0^\infty
			\sqrt{\abs{\set{i \mid \abs{v_i} \ge x}}}
			dx\ .
	\end{align}
\end{definition}
We extend the definition to matrices as
\begin{definition}
	\label[definition]{def:matrixSnorm}
	For a real $n \times n$ matrix $A$ we define $\snorm{A}$ by:
	\begin{align}
		\notag
		\snorm{A} = \sup_{u \neq 0} \set{ \frac{\snorm{Au}}{\snorm{u}} }.
	\end{align}
\end{definition}
We analyze this norm in section~\ref{sec:kwalks} showing several properties.
We use this norm to reason about the number of $k$-walks starting in a specific
set of nodes $S\subseteq V$, by showing that this number is at most
$\sqrt{\abs{S}}\snorm{X_G^k\mathbf{1}}$. The main technical lemma of the paper
is to show that if $G$ is a graph with no $2k$-cycle and maximum degree at most
$m^{2/(k+1)}$, then $\snorm{X_G} = O(k^2 m^{1/(k+1)})$.

\subsection{Related work}\label{sec:related}
All stated bounds are in the RAM model unless otherwise specified and $k$ is
assumed to be fixed. We will review related work of both given even and odd
cycles.

\paragraph{Combinatorial upper bounds.}
We briefly discuss known combinatorial bounds other than the previously mentioned \cite{YusterZ97,AlonYZ97,Monien85}.
Alon, et al. \cite{AlonYZ97} also showed several results for directed
graphs. In particular, an
upper bound of $O(m^{2-1/k})$ to find a $C_{2k}$, as well as $O(m^{2-
\frac{2}{\ell+1}})$ to find $C_\ell$ for odd $\ell$.
In the same paper, Alon, et al.~\cite{AlonYZ97} also
present bounds parameterized on the \emph{degeneracy} of the graph: the
degeneracy $d(G)$ of a graph $G$ is the largest minimal degree taken over
all the subgraphs of $G$, and for any $G$ it can be bounded from
above by $d(G) \leq 1/2 m^{1/2}$. They present bounds
of the form $O(m^{\alpha} d(G)^{\beta})$. These bounds also apply to
directed graphs. We note, that for undirected graphs the result of
Theorem~\ref{thm:2k_cycle} is still asymptotically better for $d(G) =
\omega(1)$.
The problem of combinatorially finding a $C_3$ has also been studied
thoroughly in the literature. The current fastest bound is due to
Yu~\cite{Yu2015} and uses $O(n^3 \poly(\log\log n)) / \log^4 n )$ time in the
word-RAM model with word-size $\Omega( \log n)$. For sparse graphs a folklore
$O(m^{3/2})$ algorithm exists

\paragraph{Non-combinatorial upper bounds.}
As mentioned, the best algorithm to find general cycles is due to the
seminal paper introducing color-coding, Alon et al. \cite{AlonYZ95} who gave
an $O(n^{\omega})$ expected time upper bound, and an $O(n^{\omega} \log n)$ worst
case upper bound, for finding a $C_\ell$ in a directed or undirected graph.
Other algorithms improve on \cite{AlonYZ95} for finding specific $C_\ell$.
Alon et al. \cite{AlonYZ97}
showed that a $C_3$ can be
found in time $O(m^{\frac{2\omega}{\omega -1}}) = o(m^{1.41})$ in both directed and
undirected graphs. Extending this, Eisenbrand and Grandoni
\cite{EISENBRAND200313} showed a $O(n^{1/\omega}m^{2-2/\omega})$ time upper
bound for $C_4$ in directed graphs. Both the former and the latter bounds are
asymptotically faster than $O(n^\omega)$ for sufficiently sparse input.
Improving asymptotically on Eisenbrand and
Grandoni for sparse graphs, Yuster and Zwick \cite{Yuster:2004:DSD:982792.982828} showed a
$O(m^{(4 \omega -1)/(2 \omega +1)}) = o(m^{1.48})$ upper bound for directed
graphs. For finding a $C_6$ in graphs with low degeneracy $d(G)$, Alon et
al.~\cite{AlonYZ97} showed a bound of $O((m d(G))^{2\omega/(\omega+1})) =
O((md(G))^{1.41})$.

\subsection{Notation}
Let $G=(V,E)$ be a graph. For (not necessarily disjoint) sets of nodes
$A,B\subseteq V$ we let $E(A,B)$ denote the set of edges between $A$ and $B$ in
$G$, i.e. $E\cap (A\times B)$. We use $E(v,A)$ to denote $E(\{v\},A)$.

\section{Finding even cycles}\label{sec:2kcycles}
In this section we describe our algorithm for finding a $C_{2k}$ in an undirected graph
$G=(V,E)$ with $n$ nodes and $m$ edges. In our analysis we will assume
Lemma~\ref{lem:main_kwalks}, but we defer the actual proof of the lemma to
Section~\ref{sec:kwalks}.

Our algorithm works by creating a series of graphs $G_{\le 1}^k, \ldots,
G_{\le n}^k$ guaranteed to contain any $2k$-cycle that may exist. Furthermore,
the total size of these graphs can (essentially) be bounded by the total number
of $\preceq$-capped $k$-walks which is used to bound the running time.

\begin{proof}[Proof of Theorem~\ref{thm:2k_cycle}]
    Let $A$ be any algorithm that takes a graph $H$ and a node $u$ in $H$
    as input and determines if $u$ is contained in a $2k$-cycle in time
    $O(g(k)\cdot |E(H)|)$.

    Order the nodes of $G$ as $v_1,\ldots, v_n$ non-decreasingly by degree and
    define $G_{\le i}$ to be the subgraph of $G$ induced by $v_1,\ldots, v_i$.
    Let $G_{\le i}^k$ denote the subgraph of $G_{\le i}$ containing all
    edges (and their endpoints) incident to nodes at distance $<k$ from $v_i$
    in $G_{\le i}$. Now for each $i\in \{1,\ldots, n\}$ in increasing order we
    create the graph $G_{\le i}^k$, run algorithm $A$ on $G_{\le i}^k$
    and $v_i$, and return any $2k$-cycle found (stopping the algorithm).
    If no such cycle is found for
    any $i$ the algorithm returns that no $2k$-cycle exists in $G$.

    For correctness let $C$ be any $2k$-cycle in $G$ and let $v_i$ be the
    node in $C$ that is last in the ordering. It then follows from the
    definition that $C$ is fully contained in $G_{\le i}^k$ and thus either the
    algorithm returns a $2k$-cycle when $A$ is run on $G_{\le i}^k$ or 
    some other $2k$-cycle when $A$ is run on $G_{\le j}^k$ for $j <
    i$. For the running time observe first that creating the graphs $G_{\le
    i}^k$ and running algorithm $A$ on these graphs takes time proportional to
    the total number of edges in these graphs. Thus what is left is to bound
    this number of edges. The number of edges in $G_{\le
    i}^k$ is bounded by the number of capped $k$-walks starting in $v_i$ in
    $G$. Let $i$ be the largest value such that
    $G_{\le i}^k$ does not contain a $2k$-cycle and $\deg(v_i)\le m^{2/(k+1)}$.
    It then follows by Lemma~\ref{lem:main_kwalks} that the graphs $G_{\le
    1}^k,\ldots, G_{\le i}^k$ contain at most a total number of $O(f(k)\cdot
    m^{2k/(k+1)})$ edges. Furthermore, there are at most $m^{1-2/(k+1)}$ nodes
    of degree $> m^{2/(k+1)}$, and thus the total number of edges over all the
    graphs $G_{\le 1}^k,\ldots, G_{\le n}^k$ is at most $O(f(k)\cdot
    m^{2k/(k+1)})$ giving the desired running time.
\end{proof}
As an example, the algorithm $A$ in the above proof could be the algorithm
of Monien~\cite{Monien85} or Alon et al.~\cite{AlonYZ95}.

\section{Bounding the number of capped $k$-walks}\label{sec:kwalks}
In this section we will prove Lemma~\ref{lem:main_kwalks}. Let $G=(V,E)$ be a
given graph. We will denote the nodes of $G$ by $u_1,\ldots, u_n$ or simply
$1,\ldots, n$ if it is clear from the context.

Recall the definition of $\snorm{\cdot}$ from the introduction. We note that
the following basic properties hold.
\begin{lemma}
	\label[lemma]{lem:snormBasic}
	For all vectors $u,v \in \R^n$ and $c \in \R$ we have:
	\begin{align}
		\notag
        \snorm{u+v} &\le \snorm{u} + \snorm{v},
		\\\notag
        \snorm{cu} &= \abs{c} \cdot \snorm{u},
		\\\notag
        \snorm{u} = 0 &\iff u = 0
		\, .
	\end{align}
\end{lemma}

As mentioned in the introduction, we would like to use the $\snorm{\cdot}$-norm
of $X_G$ to bound the number of $k$-walks starting in a given subset
$S\subseteq V$. We can do this using the following lemma.
\begin{lemma}\label[lemma]{lem:kwalks_set}
	Let $G = (V,E)$ be a graph with $n$ nodes and adjacency matrix $X_G$. Let
    $S \subseteq V$ be a set of nodes. For any integer $k$ the number of
    $k$-walks starting in $S$ is bounded by $\sqrt{\abs{S}}\snorm{X_G^k
    \mathbf{1}}$.
\end{lemma}
\begin{proof}
Let $v = X_G^k \mathbf{1}$ and let $w$ be the vector such that $w_i = v_i$ when
$i \in S$ and $w_i = 0$ when $i \notin S$. Then the number of $k$-walks
starting in $S$ is exactly the sum of entries in $w$, i.e. it is
$\norm{w}_1$. So the number of $k$-walks starting in $S$ is bounded by
\begin{align}
	\notag
    \norm{w}_1 &=
	\int_0^\infty
		\abs{\set{i \mid w_i \ge x}}
		dx
    \\\notag &\le
	\sqrt{\abs{S}}
	\int_0^\infty
		\sqrt{\abs{\set{i \mid w_i \ge x}}}
		dx
    \\\notag &=
	\sqrt{\abs{S}} \snorm{w}
    \\\notag &\le
	\sqrt{\abs{S}} \snorm{v}
	\, ,
\end{align}
as desired. Here the first inequality follows because $w$ has at most $|S|$
non-zero entries.
\end{proof}

To prove Lemma~\ref{lem:main_kwalks} we want to bound the quantity
$\snorm{X_G^k}$ for graphs, $G$, which do not contain a $2k$-cycle and
have maximum degree at most $m^{\frac{2}{k+1}}$.
To do this we will need the following lemmas, which are proved
in \Cref{sec:proofs}. 

\begin{lemma}
	\label[lemma]{lem:matrixSnormOnlyZeroOne}
    Let $A$ be a real $n\times n$ matrix. If, for all vectors $v\in \{0,1\}^n$
    we have $\snorm{Av} \le C\snorm{v}$ for some value $C$, then $\snorm{A}\le
    16C$.
\end{lemma}

\begin{lemma}
	\label[lemma]{lem:modifiedBS}
	Let $G$ be a graph with and let $A$ and $B$ be subsets of nodes in $G$. Let $k \ge 2$
	be an integer and assume that $G$ contains no $2k$-cycle. Then
	\begin{align}
		\abs{E(A,B)} \le
        100k \cdot \left(\sqrt{\abs{A} \cdot \abs{B}}^{1+1/k} + \abs{A} + \abs{B} \right )\, .
	\end{align}
\end{lemma}

We are now ready to prove the main technical lemma stated below.
\begin{lemma}
	\label[lemma]{lem:mainTechnical}
	Let $G = (V,E)$ be a graph with $m$ edges and let $k$ be a positive
    integer. Assume that $G$ has maximum degree at most $m^{2/(k+1)}$ and does
    not contain a $2k$-cycle. Let $X_G$ be the adjacency matrix for $G$, then
	\begin{align}
		\notag
		\snorm{X_G} = O\left(k^2 m^{1/(k+1)} \right)
		\, .
    \end{align}
\end{lemma}
\begin{proof}
    We denote the vertices of $G$ by $1,2,\ldots, n$ for convenience. By
    \Cref{lem:matrixSnormOnlyZeroOne} we only need to show that $\snorm{X_G v}
    = O\left(k^2 m^{1/(k+1)} \snorm{v}\right)$ for every vector $v$ where each
    entry is either $0$ or $1$. Each such vector, $v$, can be viewed as a set
    of nodes $A\subseteq V$, where $v_i$ is $1$ whenever $i\in A$ and $0$
    otherwise. We will adopt this view and denote $v$ by $\mathbf{1}_A$.
    In this case we have $\snorm{\mathbf{1}_A} =
    \sqrt{\abs{A}}$. Thus it suffices to show that for all $A\subseteq V$
    we have
    \begin{align}
    	\label[equation]{eq:mainTechnicalFirstRewrite}
        \snorm{X_G \mathbf{1}_A} = O \left( k^2 m^{1/(k+1)} \sqrt{A} \right )
    	\, .
    \end{align}
    Now fix an arbitrary $A \subseteq V$. We are going to show that
    \eqref{eq:mainTechnicalFirstRewrite} holds. For every non-negative integer
    $i$ we let $B_i$ denote the set of nodes in $G$ which have more than
    $2^{i-1}$ but at most $2^i$ neighbours in $A$. That is
    \begin{align}
    	\notag
    	B_i = \set{v \in V \mid \abs{E(v,A)} \in \left(2^{i-1},2^i\right] }
    	\, .
    \end{align}
    We note that by the definition of $\snorm{\cdot}$ we have that
    \begin{align}
    	\notag
        \snorm{X_G \mathbf{1}_A}
        &\le
    	\sum_{i \ge 0}
    		2^i \sqrt{\sum_{j \ge i} \abs{B_j}}
        \\\notag &\le
    	\sum_{i \ge 0}
    		2^i \sum_{j \ge i} \sqrt{\abs{B_j}}
        \\\notag &<
    	2
    	\cdot
    	\sum_{i \ge 0}
    		2^i \sqrt{\abs{B_i}}
    	\, .
    \end{align}
    So in order to show \eqref{eq:mainTechnicalFirstRewrite} it suffices to
    show \eqref{eq:mainTechnicalSecondRewrite} below
    \begin{align}
    	\label[equation]{eq:mainTechnicalSecondRewrite}
        \sum_{i\ge 0} 2^i\sqrt{\abs{B_i}} = O\!\left( k^2 m^{1/(k+1)} \sqrt{|A|} \right )\, .
    \end{align}
    or alternatively to show
    \begin{align}
    	\label[equation]{eq:mainTechnicalSecondRewrite2}
        \sum_{i\ge 0} 2^i\frac{\sqrt{\abs{B_i}}}{\sqrt{|A|}} = O\!\left( k^2 m^{1/(k+1)}\right )\, .
    \end{align}

    For an integer $i\ge 0$ let $t_i$ be defined by
    \begin{align}
    	\notag
    	t_i = 2^i\cdot \frac{\sqrt{\abs{B_i}}}{\sqrt{\abs{A}}}
    	\, .
    \end{align}
    We will bound the value $t_i$ by looking at the number of edges between the
    sets $B_i$ and $A$. Our plan is to bound the value $t_i$ in several ways,
    and then taking a geometric mean will yield the result.
    Observe first, that by the definition of $B_i$ we have
    at least $2^{i-1} \abs{B_i}$ edges from $B_i$ to $A$, and hence $2^i
    \abs{B_i} \le 2\abs{E(B_i,A)} \le 2m$. It follows that $t_i$ is bounded by
    \begin{align}
    	\notag
    	t_i
    	= \frac{2^i \sqrt{\abs{B_i}}}{\sqrt{\abs{A}}}
    	= \frac{2^{i/2} \sqrt{2^i \abs{B_i}}}{\sqrt{\abs{A}}}
    	\le \frac{2^{i/2} \sqrt{2m}}{\sqrt{\abs{A}}}
    	\, .
    \end{align}
    Let $A_i$ be the subset of nodes of $A$ that are adjacent to a node in $B_i$,
    then $E(B_i,A) = E(B_i,A_i)$. By \Cref{lem:modifiedBS} it also follows that
    \begin{align}
    	\notag
    	t_i
    	&\le \frac{2\abs{E(B_i,A_i)}}{\sqrt{\abs{B_i} \cdot \abs{A}}}
        \\\notag &\le
    	200k \sqrt{\abs{B_i} \cdot \abs{A_i}}^{1/k} +
    	200k \sqrt{\frac{\abs{B_i}}{\abs{A}}} +
    	200k \sqrt{\frac{\abs{A_i}}{\abs{B_i}}}
    	\, .
    \end{align}

    We also note that $t_i = 0$ whenever $i > d$ where $d$ is the smallest
    integer such that $2^{d-1} > m^{2/(k+1)}$, since the maximum degree of the
    graph is $m^{2/(k+1)}$. It follows that the sum $\sum_{i\ge 1}t_i$ can
    be bounded by: 
    \begin{align}
    	\notag
        &\quad O\!\left(
    	\sum_{i=1}^d
    		\min \set {
    			\frac{2^i \sqrt{\abs{B_i}}}{\sqrt{\abs{A}}},
                k \sqrt{\abs{B_i}\abs{A_i}}^{\frac{1}{k}} +
    			k\sqrt{\frac{\abs{B_i}}{\abs{A}}} +
    			k\sqrt{\frac{\abs{A_i}}{\abs{B_i}}}
    		}
        \right)
    	\\
    	\notag
    	&=
        O\!\left(
        \Sigma_1
    	+
    	\sum_{i=1}^d
    		\min \set {
    			\frac{2^i \sqrt{\abs{B_i}}}{\sqrt{\abs{A}}},
    			k\sqrt{\frac{\abs{B_i}}{\abs{A}}} +
    			k\sqrt{\frac{\abs{A_i}}{\abs{B_i}}}
    		}
            \right)
    	\\
    	&=
        \label{eq:s1_s2}
        O\!\left(
        \Sigma_1
    	+
    	\sum_{i=1}^d
        \left(
    		k\sqrt{\frac{\abs{B_i}}{\abs{A}}} +
    		k\cdot2^{i/2}
            \right)
        \right)\, 
    \end{align}
    where
    \[
        \Sigma_1 =
    	\sum_{i=1}^d
    		\min \set {
    			\frac{2^{i/2} \sqrt{2m}}{\sqrt{\abs{A}}},
    			k\sqrt{\abs{B_i} \cdot \abs{A}}^{1/k}
    		}
    \]
    Here, we have $\sqrt{\frac{\abs{A_i}}{\abs{B_i}}} \le 2^{i/2}$ because each
    node of $B_i$ has at most $2^i$ neighbours in $A$.


    Let $\Sigma_1$ and $\Sigma_2$ denote the two sums of \eqref{eq:s1_s2}
    above respectively. We will start by bounding $\Sigma_2$. Since, by
    definition, every node in $B_i$ has at least $2^{i-1}$ neighbours in $A_i$
    and every node in $A_i$ has degree at most $m^{2/(k+1)}$ we see that
    $\abs{B_i}2^{i-1} \le \abs{A_i}m^{2/(k+1)}$. Hence we get that:
    \begin{align}
        \notag
        \Sigma_2 \le
        \sum_{i=1}^d\left(
            km^{1/(k+1)}2^{(1-i)/2} +
            k2^{i/2}\right)
        =
        O \left ( km^{1/(k+1)} \right )\, .
    \end{align}
    Now we will bound $\Sigma_1$. First we note that $\abs{B_i}2^{i-1} \le m$ and
    therefore $\abs{B_i} \le \frac{2m}{2^i}$. Inserting this gives us:
    \begin{align}
        \notag
        \Sigma_1
        \le
        \sum_{i=1}^d
            \min \set {
                \frac{2^{i/2} \sqrt{2m}}{\sqrt{\abs{A}}},
                k\sqrt{\frac{2m}{2^i} \cdot \abs{A}}^{1/k}
            }\, .
    \end{align}
    Let $d_0$ be the largest integer such that $2^{d_0} \le
    \frac{\abs{A}}{(2m)^{(k-1)/(k+1)}}$.
    Then:
    \begin{align}
        \notag
        \frac{2^{d_0/2} \sqrt{2m}}{\sqrt{\abs{A}}}
        &
        =
        \Theta \left ( m^{1/(k+1)} \right )
        \\
        \notag
        \sqrt{\frac{2m}{2^{d_0}} \cdot \abs{A}}^{1/k}
        &
        =
        \Theta \left ( m^{1/(k+1)} \right )\, .
    \end{align}
    Inserting this gives us:
    \begin{align}
        \notag
        \Sigma_1
        & \le
        k
        \sum_{i=1}^d
            \min \set {
                \frac{2^{i/2} \sqrt{2m}}{\sqrt{\abs{A}}},
                \sqrt{\frac{2m}{2^i} \cdot \abs{A}}^{1/k}
            }
        \\
        & \le
        k
        \sum_{i=-\infty}^{\infty}
            \min \set {
                \frac{2^{i/2} \sqrt{2m}}{\sqrt{\abs{A}}},
                \sqrt{\frac{2m}{2^i} \cdot \abs{A}}^{1/k}
            }
        \\
        & \le
        k
        \sum_{i=-\infty}^{d_0}
            \frac{2^{i/2} \sqrt{2m}}{\sqrt{\abs{A}}}
        +
        k
        \sum_{i=d_0}^{\infty}
            \sqrt{\frac{2m}{2^i} \cdot \abs{A}}^{1/k}
        \\
        & =
        O\left ( k m^{1/(k+1)} \right )
        \cdot
        \left (
            \sum_{i = 0}^{\infty} 2^{-i/2}
            +
            \sum_{i = 0}^{\infty} 2^{-i/k}
        \right )
        \\
        & =
        O \left ( k^2 m^{1/(k+1)} \right )\, .
    \end{align}
    Summarizing, we thus have that
    \[
        \sum_{i\ge 0} t_i = O\!\left(k^2\cdot m^{\frac{1}{k+1}}\right)\ ,
    \]
    and combining this with \eqref{eq:mainTechnicalSecondRewrite2},
    \eqref{eq:mainTechnicalSecondRewrite} and
    \eqref{eq:mainTechnicalFirstRewrite} now gives us the lemma.
\end{proof}

Using Lemma~\ref{lem:mainTechnical} above we are now ready to prove
Lemma~\ref{lem:main_kwalks} which we used to bound the number of
$\preceq$-capped $k$-walks in Section~\ref{sec:2kcycles}. The main idea in the
proof of Lemma~\ref{lem:main_kwalks} is to split the nodes $V$ into different
sets based on their degrees and then use Lemma~\ref{lem:mainTechnical} to bound
the $\snorm{\cdot}$-norm of the graphs induced by these sets individually.

\begin{proof}[Proof of Lemma~\ref{lem:main_kwalks}]
    Let $V_i$ be the set of nodes $u$ with $\deg(u) \in
    \left(2^{i-1},2^i\right]$, and let $V_{\le i} = \cup_{j \le i} V_j$ be the
    set of nodes with $\deg(u) \in (0,2^i]$. Let $G_{\le i} = (V, E \cap V_{\le
    i}^2)$ be the subgraph of $G$ induced by
    $V_{\le i}$. Note that $G_{\le i}$ here is defined
    slightly differently than we did in Section~\ref{sec:2kcycles} as we
    consider entire sets of nodes $V_i$. Any $\preceq$-capped $k$-walk starting in
    from a node $u \in V_{i}$ is contained in $X_{G_{\le i}}$. It follows by
    Lemma~\ref{lem:kwalks_set} that the total number of $\preceq$-capped
    $k$-walks in $G$ is bounded by
    \begin{align}
        \notag
        \sum_{i \ge 0}
            \sqrt{\abs{V_i}} \snorm{X_{G_i}^k \mathbf{1}}
        &\le
        \sum_{i \ge 0}
            \snorm{X_{G_i}}^{k-1}
            \sqrt{\abs{V_i}} \snorm{X_{G_i} \mathbf{1}}
        \\
        \label{eq:kwalks_bound}
        &\le
        \snorm{X_G}^{k-1}
        \sum_{i \ge 0}
            \sqrt{\abs{V_i}} \snorm{X_{G_i} \mathbf{1}}
        \, .
    \end{align}
        We note that $X_{G_i} \mathbf{1} \le \sum_{j \le i} 2^j \mathbf{1}_{V_j}$, and hence
    \begin{align}
        \notag
        \sum_{i \ge 0}
            \sqrt{\abs{V_i}} \snorm{X_{G_i} \mathbf{1}}
        &\le
        \sum_{i \ge 0}
            \sqrt{\abs{V_i}} \sum_{j \le i} \snorm{2^j \mathbf{1}_{V_j}}
        \\\notag &=
        \sum_{i \ge j \ge 0}
            \sqrt{\abs{V_i}}
            \cdot \sqrt{\abs{V_j}}
            \cdot 2^j\, .
    \end{align}
    We now note that
    \begin{align}
        \notag
        \sqrt{\abs{V_i}}
        \cdot \sqrt{\abs{V_j}}
        \cdot 2^j
        &=
        \sqrt{2^i\abs{V_i}}
        \cdot \sqrt{2^j\abs{V_j}}
        \cdot 2^{-(i-j)/2}
        \\\notag &\le
        \frac{2^i\abs{V_i} + 2^j\abs{V_j}}{2}
        \cdot
        2^{-(i-j)/2}
        \, ,
    \end{align}
    which implies that
    \begin{align}
        \notag
        \sum_{i \ge j \ge 0}
            \sqrt{\abs{V_i}}
            \cdot \sqrt{\abs{V_j}}
            \cdot 2^j
        &\le
        \sum_{i \ge j \ge 0}
            \frac{2^i\abs{V_i} + 2^j\abs{V_j}}{2}
            \cdot
            2^{-(i-j)/2}
        \\\notag &=
        \sum_{i \ge 0}
            2^i \abs{V_i}
            \sum_{\ell \ge 0} 2^{-\ell/2}
        \\\notag &=
        \frac{\sqrt{2}}{\sqrt{2}-1}
        \sum_{i \ge 0}
            2^i \abs{V_i}
        \, .
    \end{align}
    Since $\sum_{i \ge 0} 2^i \abs{V_i}$ is at most twice as large as the sum of
    degrees of the nodes in $G$ it is bounded by $4m$, and therefore
    \begin{align}
        \sum_{i \ge j \ge 0}
            \sqrt{\abs{V_i}}
            \cdot \sqrt{\abs{V_j}}
            \cdot 2^j
        \le
        4 \cdot \frac{\sqrt{2}}{\sqrt{2}-1} m
        <
        14m
        \, .
    \end{align}
    Combining this with \eqref{eq:kwalks_bound} and
    Lemma~\ref{lem:mainTechnical} we get that the number of $\preceq$-capped
    $k$-walks is at most
    \begin{align}
        \notag
        14 \snorm{X_G}^{k-1} m = O\!\left((k^2)^{k-1}m^{\frac{2k}{k+1}}\right)
        \, ,
    \end{align}
    which is what we wanted to show.
\end{proof}

Below we prove Lemma~\ref{lem:kwalks_lower}, which gives a lower bound on the number of capped $k$-walks.

\begin{proof}[Proof of Lemma~\ref{lem:kwalks_lower}]
Let $\Delta = \frac{m}{2n}$.
For a subgraph $F$ of $G$ we let $f(F)$ denote the subgraph $F'$ of $F$
obtained in the following way. Initially we let $F' = F$. As long as there
exists a node $v \in F'$ such that $\deg_{F'}(v) < \Delta$ we remove $v$
from $F'$. We continue this process until no node in $F'$ has fewer than
$\Delta$ neighbours and let $f(F) = F'$.

We now construct the sequences $(H_i)_{i \ge 0}, (H_i')_{i \ge 0}$ of subgraphs
of $G$ in the following manner. We let $H_0' = G$, and $H_0 = f(H_0')$. If $H_i$
is non-empty, let $v_i$ be the largest element in $H_i$, i.e. $v_i \succeq v$
for all $v \in H_i$, and define $H_{i+1}' = H_i \setminus \set{v_i}$. If $H_i$
is empty we let $H_{i+1}' = H_i$. In either case we let $H_{i+1} = f(H_{i+1}')$.

For all $i$ such that $H_i$ is non-empty, there exists at least \linebreak$\deg_{H_i}(v_i) \Delta^{k-1}$
capped $k$-walks $(x_1,\ldots,x_k)$ with $x_1 = v_i$. By the definition of $H_{i+1}'$
we have that $\deg_{H_i}(v_i) = \abs{E(H_i)} - \abs{E(H_{i+1}')}$. The total number
of capped $k$-walks in $G$ is therefore at least:
\begin{align}
	\label{eq:cappedKWalksNumberOfWalks}
	\sum_{i \ge 0}
		\left ( \abs{E(H_i)} - \abs{E(H_{i+1}')} \right )
		\Delta^{k-1}\, .
\end{align}
Now note that:
\begin{align}
    \notag
    &\quad  \sum_{i \ge 0}
		\left ( \abs{E(H_i)} - \abs{E(H_{i+1}')} \right )
    \\
	\label{eq:cappedKWalksSumOfEdges}
    &=
	\left (
		\sum_{i \ge 0}
			\abs{E(H_i')} - \abs{E(H_{i+1}')}
	\right )
	-
	\left (
		\sum_{i \ge 0}
			\abs{E(H_i')} - \abs{E(H_{i})}
	\right )
	\, .
\end{align}
The first sum on the right hand side of \eqref{eq:cappedKWalksSumOfEdges} is a
telescoping sum that is equal to $m$. The second sum on the right hand side of
\eqref{eq:cappedKWalksSumOfEdges} can be bounded by noting that
$\abs{E(H_i')} - \abs{E(H_{i})}$ is at most $\Delta \cdot \abs{V(H_i' \setminus f(H_i'))}$,
since applying $f$ to $H_i'$ removes $\abs{V(H_i' \setminus f(H_i'))}$ nodes, and each node
removed had degree at most $\Delta$.
Since at most $n$ nodes are removed in total the sum is bounded by $n\Delta$. Hence
\eqref{eq:cappedKWalksSumOfEdges} is at least $m-n\Delta = \frac{m}{2}$.
Inserting this into \eqref{eq:cappedKWalksNumberOfWalks} gives that the number of
capped $k$-walks is at least
\begin{align}
	\notag
	\frac{m}{2} \cdot \Delta^{k-1} =
	n \cdot \left ( \frac{m}{2n} \right )^k
	\, ,
\end{align}
as desired.
\end{proof}

\section{Hardness of finding cycles}\label{sec:clb}
Theorem~\ref{thm:2k_cycle} presents an algorithm with a seemingly natural
running time in terms of $m$ and $k$. A natural question to ask is
whether the exponent of $m$ has to increase with $k$ and,
perhaps more interestingly, what the correct
exponent is. In this section we address the possibility of faster
algorithms, by proving \Cref{thm:n2opt,thm:clbs} discussed in the
introduction.

\begin{proof}[Proof of \Cref{thm:clbs}]
    Let $G=(V,E)$ be the graph in which we wish to find a triangle with $|V| =
    n$ and $|E| = \Theta(n^2)$. By Conjecture~\ref{conj:comb_bmm} it takes
    $n^{3-o(1)}$ to find a triangle in $G$. Now create the graph $G'$
    consisting of three copies, $A$, $B$, and $C$, of $V$. Denote each copy of
    $u\in V$ in $A,B,C$ by $u_A,u_B,u_C$, respectively. For each edge $(u,v)\in
    E$ add the edges $(u_A,v_B)$, $(u_B,v_C)$, and $(u_C,v_A)$ to $G'$. It now
    follows that $G$ contains a triangle $u,v,w$ if and only if $G'$ contains a
    triangle $u_A,v_B,w_C$.

    Now Fix $x = \lceil (2k+1)/4 \rceil$ and note that $2k\ge 3x$ by the
    restrictions to $k$. Create the graph $G^e_k$ by taking a copy of $G'$ and
    performing the following changes: Replace each edge by a path of length
    $x$. If $2k > 3x$ replace each node $u_A$ in $G^e_k$ by a path
    $u_A^1,\ldots, u_A^{2k-3x+1}$. Otherwise if $2k=3x$ do nothing. We now
    claim that $G^e_k$ contains a $C_{2k}$ if and only if $G$ contains a
    triangle. Observe first, that if $G$ contains a triangle $u,v,w$, then
    $u_A^1\leadsto v_V\leadsto w_C\leadsto u_A^{2k-3x+1}\leadsto u_A^1$ is a
    cycle in $G^e_k$ and has length $3x + 2k-3x = 2k$. Now assume that $G^e_k$
    has a cycle of length $2k$. If this cycle contains two nodes $u_A^1$ and
    $v_A^1$ it must have length at least $4x > 2k$ and similar for $B$ and $C$
    and $u_A^{2k-3x+1}$ and $v_A^{2k-3x+1}$. Thus, the cycle must exactly be of
    the form $u_A^1\leadsto v_V\leadsto w_C\leadsto u_A^{2k-3x+1}\leadsto u_A^1$
    and such a cycle can only have length $2k$ if all edges $(u_A,v_B)$,
    $(v_B,w_C)$, and $(w_C,u_A)$ are present in $G'$. Now observe that for
    constant $k$ the graph $G^e_k$ has $N = \Theta(n^2)$ nodes and $M =
    \Theta(n^2)$ edges. It now follows from Conjecture~\ref{conj:comb_bmm} that
    no algorithm can detect a $C_{2k}$ in $G^e_k$ in time $O(M^{3/2-\eps}) =
    O(n^{3-\eps})$ for any $\eps > 0$.
\end{proof}
The reduction for Proposition~\ref{thm:clbs} is shown in Figure~\ref{fig:clbs}
below.
\begin{figure}[htbp]
    \centering
    \includegraphics[width=0.6\columnwidth]{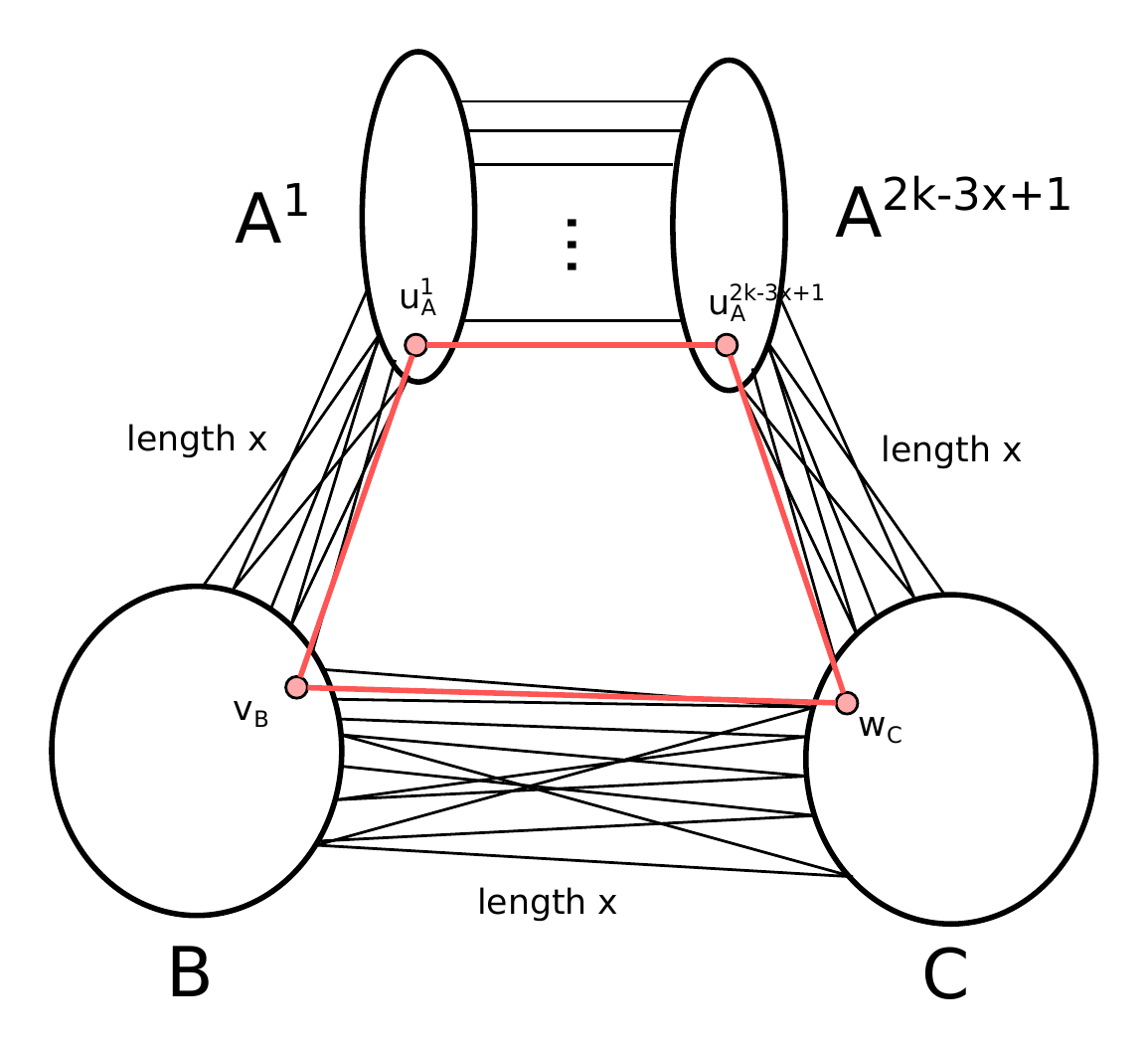}
    \caption{The construction of $G^e_k$ from the proof of Lemma~\ref{thm:clbs}
    and an example $2k$-cycle highlighted in red.}
    \label[figure]{fig:clbs}
\end{figure}

Finally, We show the ``conditional optimality'' stated in \Cref{thm:n2opt}. The theorem states that if $O(n^2)$ time is optimal, then our bound is the best that can be achieved.
\begin{proof}[Proof of \Cref{thm:n2opt}]
	Let $\eps > 0$ be given and let $\delta = \eps$.
	
	Assume there exists an algorithm
	which finds a $2k$-cycle in time $O(m^{2k/(k+1)-\eps})$. Now consider the
    following algorithm: If $m\ge 100k\cdot n^{1+1/k}$ answer yes, and
    otherwise run the given algorithm. This algorithm has running time
    $O(n^{(1+1/k)\cdot(2k/(k+1)-\eps)}) = o(n^{2-\delta})$.
    Hence part (1) holds.
    
    Now assume there are finitely many graphs $G$
    such that $\abs{E(G)} \ge \abs{V(G)}^{1+1/k-\eps}$. Then there must exist some constant $n_0$ such
    that no graph with $n\ge n_0$ nodes and $m\ge n^{1+1/k-\eps}$ edges
    contains a $2k$-cycle. Now consider the following algorithm: Let $G=(V,E)$
    be the graph we wish to detect a $C_{2k}$ in. If $|V| < n_0$ we can answer in
    constant time. If $|V| \ge n_0$ and $|E| \ge |V|^{1+1/k-\eps}$ answer no,
    and otherwise run the algorithm of \Cref{thm:2k_cycle} to detect a $C_{2k}$
    in time $O(|V|^{(1+1/k-\eps)\cdot 2k/(k+1)}) = o(|V|^{2-\delta})$.
    Hence part (2) holds.
\end{proof}

\section{Omitted proofs}\label{sec:proofs} 
This section contains missing proofs from Section~\ref{sec:kwalks}.

\begin{proof}[Proof of \Cref{lem:matrixSnormOnlyZeroOne}]
Let $v \in \R^n$ be a vector such that each entry either is contained in $\left[2^{-1},1\right]$
or is $0$. Let $r = \abs{\supp(v)}$ and write $v$ as $v = \sum_{i=1}^r \lambda_i e_i$ for vectors
$e_i$ such that for each $e_i$ there is a single entry $(e_i)_j = 1$ and all other entries
are $0$.
Let $X_1,\ldots,X_r$ be independent random variables $\in \set{0,1}$ such that $E(X_i) = \lambda_i$.
By the concavity of $\snorm{\cdot}$ we then have
\begin{align}
    \notag
	\snorm{Av}
    &=
	\snorm{E \left ( A \sum_{i=1}^r X_i e_i\right )}
    \le
	E \left ( \snorm{A \sum_{i=1}^r X_i e_i} \right )
    \\\notag &\le
	E \left ( C\snorm{\sum_{i=1}^r X_i e_i} \right )
    \le
	C\sqrt{r}
    \\
	\label{eq:vectorHalfToOneBound}
    &\le
	2C\snorm{v}
	\, .
\end{align}
Since $v$ was arbitrarily chosen \eqref{eq:vectorHalfToOneBound} holds for all vector
$v$ with entries in $\set{0} \cup [2^{-1},1]$.

Let $v \in \R^n$ be a vector where each entry is non-negative. We will show that
$\snorm{Av} \le 8C \snorm{v}$. For each integer $k$ let $v^{(k)} \in \R^n$ be the vector
containing the $i$'th entry of $v_i$ if $v_i \in (2^{k-1},2^k]$ and $0$ otherwise, i.e.
\begin{align}
	\notag
	v^{(k)}_i = \left [ v_i \in (2^{k-1},2^k] \right ] v_i
	\, .
\end{align}
Using the triangle inequality and \eqref{eq:vectorHalfToOneBound} on the vectors
$2^{-k}v^{(k)}$ now gives us
\begin{align}
    \notag
	\snorm{Av} =
	\snorm{\sum_{k} Av^{(k)}}
    &\le
	\sum_{k} 2^k \snorm{A 2^{-k}v^{(k)}}
    \\\notag &\le
	\sum_{k} 2^k \cdot 2C\snorm{2^{-k} v^{(k)}}
	\\ \label{eq:AvLessThanSum}
    &=
	2C\sum_{k} \snorm{v^{(k)}}
	\, .
\end{align}
Now we have that
\begin{align}
	\notag
	\sum_{k} \snorm{v^{(k)}}
    &=
	\sum_{k}
		\int_0^{2^k}
			\sqrt{\abs{\set{i \mid v^{(k)}_i \ge x}}}
			dx
	\\\notag & \le
	\sum_{k}
		2^k
		\sqrt{\abs{\set{i \mid v^{(k)}_i \ge 2^{k-1}}}}
	\\
	\notag
	& =
	4
	\sum_{k}
		\int_{2^{k-2}}^{2^{k-1}}
			\sqrt{\abs{\set{i \mid v^{(k)}_i \ge x}}}
			dx
	\\
	\label{eq:sumLessThanv}
	&
	\le
	4
	\sum_{k}
		\int_{2^{k-2}}^{2^{k-1}}
			\sqrt{\abs{\set{i \mid v_i \ge x}}}
			dx
	=
	4\snorm{v}\, .
\end{align}
Combining \eqref{eq:AvLessThanSum} and \eqref{eq:sumLessThanv} gives that $\snorm{Av} \le 8C\snorm{v}$
for every non-negative vector $v \in \R^n$ as desired.

Let $v \in \R^n$ be any real vector. Let $v^+$ and $v^-$ be defined by
\begin{align}
	\notag
	(v^+)_i = \max \set{v_i,0},
	\ \
	(v^-)_i = \max \set{-v_i,0}\, .
\end{align}
Then $v^+$ and $v^-$ have non-negative coordinates and $v = v^+ - v^-$. It is easy
to see that $\snorm{v} \ge \max\set{\snorm{v^+},\snorm{v^-}}$, and therefore:
$\snorm{v^+} + \snorm{v^-} \le 2\snorm{v}$. Now we get the result by the using the triangle
inequality:
\begin{align}
	\notag
    \snorm{Av} &=
	\snorm{Av^{+} - Av^{-}}
    \\\notag &\le
	\snorm{Av^{+}} + \snorm{Av^{-}}
    \\\notag &\le
	8C \left ( \snorm{v^+} + \snorm{v^{-}} \right )
    \\\notag &\le
	16C \snorm{v}
	\, .
\end{align}
It follows that $\snorm{A} \le 16C$.
\end{proof}

Below we show \Cref{lem:modifiedBS}, which can be seen as a modified version of the classic Bondy and Simonovits lemma, as we here argue about edges between any two subsets of the graph, instead of edges in the entire graph as in the original lemma \cite{BONDY197497}.
\begin{proof}[Proof of \Cref{lem:modifiedBS}]
    Let $m=|E(A,B)|$ and let $E = E(A,B)$.
    We will assume that $m \ge 100k\cdot(|A| + |B|)$ as the statement is
    otherwise trivially true. We will assume that the graph contains no
    $2k$-cycle and show that then $m\le 100k\cdot \sqrt{|A|+|B|}^{1+1/k}$.

    Let $2\alpha = \frac{m}{|A|}$ and let $2\beta =
    \frac{m}{|B|}$ be the average degrees of nodes in $A$ and $B$
    respectively when restricted to $E$. Recursively remove any node from
    $A$ respectively $B$ which does not have at least $\alpha$ respectively
    $\beta$ edges in $E$. Then we remove strictly less than $\alpha\cdot
    |A| + \beta\cdot |B| < m$ edges and thus have a non-empty graph
    left.

    Now fix some node $u\in A$ and let $L_0 = \{u\}$. Now define $L_{i+1}$ to be
    the neighbours of the nodes in $L_i$ using the edges of $E$ for
    $i=0,\ldots, k-1$. This gives us the sets $L_0,\ldots, L_k$.
    Note that if $A\cap B = \emptyset$ we have $L_i\cap L_{i+1} =
    \emptyset$ for each $i = 0,\ldots, k-1$. We will show by induction that
    $|L_i| \le |L_{i+1}|$ for each $i=0,\ldots, k-1$. This is clearly true for
    $i=0$ since $u$ has degree at least $\alpha\ge 50k$ by assumption. Now fix
    some $i \ge 1$ and assume that the statement is true for all $j < i$. We will
    assume that $i$ is even (the other case is symmetric). We know
    from~\cite{BONDY197497,YusterZ97} that
    \[
        |E(L_i,L_{i+1})| \le 4k\cdot (|L_i| + |L_{i+1}|)\ ,
    \]
    as otherwise we can find a $2k$-cycle. By the induction hypothesis this
    gives us
    \[
        |E(L_{i-1},L_i)| \le 8k\cdot |L_i|\ .
    \]
    Since $i$ is even we also know that
    \[
        \alpha\cdot |L_i| \le |E(L_{i-1},L_i)| + |E(L_i,L_{i+1})|\ ,
    \]
    and thus
    \[
        (\alpha - 8k)\cdot |L_i| \le |E(L_i,L_{i+1})| \le 4k\cdot (|L_i| +
        |L_{i+1}|)\ .
    \]
    This gives us that $(\alpha-12k) \le 4k\cdot |L_{i+1}|$, and it follows
    that
    \[
        |L_{i+1}| \ge \frac{\alpha-12k}{4k}\cdot |L_i|\ .
    \]
    By our assumption on $\alpha$ this proves that $|L_{i+1}| \ge L_i$. When
    $i$ is odd the same argument gives us that $|L_{i+1}| \ge
    \frac{\beta-12k}{4k}\cdot |L_i|$.

    By the above discussion it follows that
    \begin{align*}
        |L_k| &\ge \left(\frac{\alpha -
        12k}{4k}\right)^{\ceil{k/2}}\cdot\left(\frac{\beta-12k}{4k}\right)^{\floor{k/2}}
        \\ &\ge \frac{\alpha^{\ceil{k/2}}\beta^{\floor{k/2}}}{(8k)^k}\ ,
    \end{align*}
    where the last inequality follows by our assumption the $\alpha,\beta \ge
    50k$. Assume now that $k$ is odd (as the even case is handled similar). It
    then follows that
    \[
        |B| \ge |L_k|
        \ge \frac{\alpha^{\ceil{k/2}}\beta^{\floor{k/2}}}{(8k)^k}\ ,
    \]
    and a symmetric argument gives us
    \[
        |A| \ge \frac{\alpha^{\floor{k/2}}\beta^{\ceil{k/2}}}{(8k)^k}\ ,
    \]
    implying that
    \[
        \sqrt{|A|\cdot |B|} \ge \frac{\sqrt{\alpha\beta}^k}{(8k)^k}
        = \frac{\sqrt{\frac{m^2}{4|A||B|}}^k}{(8k)^k}\ .
    \]
    Now taking the $k$th root and isolating $m$ yields exactly the bound we
    wanted to show
    \[
        m\le 16k\cdot\sqrt{|A||B|}^{1+1/k}\ .
    \]

    In the above proof we assumed that $A$ and $B$ were disjoint in order to
    apply the lemma of~\cite{BONDY197497,YusterZ97}. Now observe that if this
    is not the case we can pick subsets $A'\subseteq A$ and $B'\subseteq B$
    with $A'\cap B' = \emptyset$ and $E(A',B') \ge m/2$ and the argument now
    follows through. 
\end{proof}

\bibliographystyle{plain}
\bibliography{cycles}

\end{document}